\documentclass{article}
\usepackage[english]{babel}
\usepackage[boxed]{algorithm}
\usepackage[noend]{algorithmic}
\usepackage{amsmath,amssymb,amsfonts,amsthm}
\usepackage[fleqn,tbtags]{mathtools}
\usepackage{lastpage}
\usepackage{bbm}
\usepackage{fancybox}
\usepackage{fullpage}
\usepackage{boxedminipage}
\usepackage{fancyhdr}
\usepackage{graphicx}
\usepackage{subfigure}
\usepackage{epsfig}
\usepackage{color}
\usepackage{xspace}
\usepackage{authblk}
\usepackage{url}
\usepackage [autostyle, english = american]{csquotes}
\usepackage{ifpdf}

\usepackage[textsize=scriptsize]{todonotes}

\ifpdf    
\usepackage{hyperref}
\else    
\usepackage[hypertex]{hyperref}
\fi

\MakeOuterQuote{"}

\newtheorem{theorem}{Theorem}
\newtheorem{lemma}[theorem]{Lemma}

\newtheorem{definition}{Definition}

\newtheorem{claim}[theorem]{Claim}

\newtheorem{observation}[theorem]{Observation}

\newtheorem{notation}{Notation}

\makeatletter
\newtheorem*{rep@theorem}{\rep@title}
\newcommand{\newreptheorem}[2]{%
\newenvironment{rep#1}[1]{%
 \def\rep@title{#2 \ref{##1}}%
 \begin{rep@theorem}}%
 {\end{rep@theorem}}}
\makeatother

\newreptheorem{theorem}{Theorem}

\newcommand {\ignore} [1] {}

\DeclareMathOperator{\supp}{supp}

\newcommand{\R}{\mathbb{R}}

\newcommand{\Y}{\mathcal{Y}}

\newcommand{\eps}{\varepsilon}

\newcommand{\Norm}{\mathcal{N}}

\newcommand{\enet}{{\cal N}_{\frac{1}{2}}}
\newcommand{\fnet}{{\cal N}_{\frac{1}{4}}}

\newcommand{\dotp}[2]{\left\langle #1 , #2 \right\rangle}
\renewcommand{\log}{\lg}

\newtheorem{innercustomlem}{Lemma}
\newenvironment{customlem}[1]
  {\renewcommand\theinnercustomlem{#1}\innercustomlem}
  {\endinnercustomlem}

 \author{
 Mikael M\o ller H\o gsgaard \thanks{Computer Science Department. Aarhus University. \texttt{hogsgaard@cs.au.dk}.} \qquad  
 Lior Kamma \thanks{School of Computer Science. Academic College of Tel-Aviv Yaffo. \texttt{liorkm@mta.ac.il}.}\qquad  
 Kasper Green Larsen \thanks{Computer Science Department. Aarhus University. \texttt{larsen@cs.au.dk}.} \qquad  
 Jelani Nelson \thanks{Department of EECS. UC Berkeley. \texttt{minilek@berkeley.edu}. Supported by NSF grant CCF-1951384, ONR grant N00014-18-1-2562, and
 ONR DORECG award N00014-17-1-2127}\and 
 Chris Schwiegelshohn \thanks{Computer Science Department. Aarhus University. \texttt{schwiegelshohn@cs.au.dk}.}}

\title{Sparse Dimensionality Reduction Revisited}

\begin{document}

\date{}
\maketitle

\begin{abstract}
  The sparse Johnson-Lindenstrauss transform is one of the central
  techniques in dimensionality reduction. It supports embedding a set
  of $n$ points in $\mathbb{R}^d$ into $m=O(\varepsilon^{-2} \lg n)$
  dimensions while preserving all pairwise
  distances to within $1 \pm \eps$. Each input point $x$ is embedded
  to $Ax$, where $A$ is an $m \times d$ matrix having $s$ non-zeros
  per column, allowing for an embedding time of $O(s \|x\|_0)$.

  Since the sparsity of $A$ governs the embedding time, much work has
  gone into improving the sparsity $s$. The current state-of-the-art
  by Kane and Nelson (JACM'14)
  shows that $s = O(\eps^{-1} \lg n)$ suffices. This is almost matched
  by a lower bound of $s = \Omega(\eps^{-1} \lg n/\lg(1/\eps))$ by
  Nelson and Nguyen (STOC'13). Previous work thus suggests that we have
  near-optimal embeddings.

  In this work, we revisit sparse embeddings and identify a loophole
  in the lower bound. Concretely, it requires $d \geq n$, which in
  many applications is unrealistic. We exploit this loophole to give a
  sparser embedding when $d = o(n)$, achieving $s = O(\eps^{-1}(\lg
  n/\lg(1/\eps)+\lg^{2/3}n \lg^{1/3} d))$. We also complement our
  analysis by strengthening the lower bound of Nelson and Nguyen to
  hold also when $d \ll
  n$, thereby matching the first term in our new sparsity upper bound. Finally, we also improve the sparsity of the best oblivious
  subspace embeddings for optimal embedding dimensionality.
\end{abstract}

\thispagestyle{empty}
\newpage
\setcounter{page}{1}

\section{Introduction}
{\em Dimensionality reduction} is a central technique for speeding up algorithms for large scale data analysis and reducing memory consumption for storage. 
A Euclidean-distance-preserving dimensionality reduction is, loosely speaking, an embedding of a high-dimensional Euclidean space into a space of low dimension, that approximately preserves the Euclidean distance between every two points. 
One of the cornerstone results is the Johnson-Lindenstrauss transform~\cite{JL84}, stating that every set of $n$ points in $d$-dimensional space can be embedded into only $m=O(\eps^{-2} \lg n)$ dimensions while preserving all pairwise Euclidian distances between points to within a factor $(1\pm \eps)$. The simplest (random) constructions of such dimensionality reducing maps, known as the {\em Distributional Johnson-Lindenstrauss Lemma}, samples a random $m \times d$ matrix $A$ with entries either i.i.d. $\Norm(0,1)$ distributed or as uniform Rademachers ($-1$ or $1$ with probability $1/2$ each). For a set $X \subset \R^d$ of $n$ points, it then holds with probability at least $1-1/n$ that $L = A/\sqrt{m}$ satisfies
\begin{eqnarray}
\label{eq:jl}
  \forall x,y \in X : \|Lx-Ly\|_2^2 \in (1 \pm \eps)\|x-y\|_2^2 \; .
\end{eqnarray}
We say that a matrix $L$ satisfying~\eqref{eq:jl} is an {\em $\eps$-JL matrix for $X$}. It is worth noting that some works require that an $\eps$-JL matrix satisfies~\eqref{eq:jl} without the square on the Euclidian norm. The two definitions are equivalent up to a constant factor scaling in $\eps$ and we work with the former as it simplifies calculations.

While the target dimension of $m = O(\eps^{-2} \lg n)$ is known to be optimal~\cite{JW13, LN17}, even when $d = O(m)$, computing the embedding $Lx$ of a point $x$ using the construction above requires $\Omega(m d) = \Omega(\eps^{-2}d \lg n)$ operations. In some applications, this may constitute the computation bottleneck, hence much work has gone into designing faster embedding algorithms. These works may roughly be categorized by two approaches. (1) Constructions that use structured embedding matrices with fast matrix-vector multiplication algorithms; and (2) constructions using sparse embedding matrices.

A classic example of the former approach is the FastJL transform by Ailon and Chazelle~\cite{AC09}. Their construction embeds a point $x$ by computing the product $PHDx$, where $D$ is a diagonal matrix with random signs on the diagonal, $H$ is a $d \times d$ Hadamard matrix and $P$ is a random sparse matrix where each entry is non-zero only with some small probability. The main idea in that construction is that $HD$ ``spreads'' the mass of the vector $x$ evenly among its coordinates, which allows for a very sparse $m \times d$ embedding matrix $P$. In addition, Hadamard matrix has an $O(d \lg d)$ matrix-vector multiplication algorithm. Analyzing the FastJL transform, and specifically the correct tradeoff between the target dimension and embedding time, has been studied extensively (see e.g. \cite{DGYNT09, KW11, FL20,JPS20}). The state-of-the-art tight analysis by  Fandina, H{\o}gsgaard and Larsen~\cite{FML22} shows that the embedding time can be bounded by $O(d \lg d + \min\{\eps^{-1} d \lg n, m \lg n \cdot \max\{1, \eps \lg n/\lg(1/\eps)\}\})$.

In the latter approach one instead designs embedding matrices with only $s \ll m$ non-zeros per column. Given such a sparse embedding matrix, it is straightforward to embed a point in $O(d s)$ time instead of $O(dm)$, hence minimizing $s$ has been the focus of extensive work. The current sparsest embedding construction is due to Kane and Nelson~\cite{KN14}, achieving a sparsity upper bound of $s = O(\eps^{-1} \lg n)$. Nelson and Nguyen~\cite{NN13} presented a lower bound of $s = \Omega(\eps^{-1} \lg n/\lg(m/\lg n))$ for any sparse $\eps$-JL matrix, almost settling the optimality of the construction by Kane and Nelson. For optimal target dimension $m = \Theta(\eps^{-2} \lg n)$, this simplifies to $s = \Omega(\eps^{-1} \lg n/\lg(1/\eps))$. While $O(d \eps^{-1} \lg n/\lg(1/\eps))$ is often larger than the near $O(d \lg d)$ embedding time achieved by FastJL, sparse embeddings have one significant advantage in that they may also exploit sparsity in the input points. Concretely, the embedding time of a point $x$ is easily seen to be $O(s\|x\|_0 )$, where $\|x\|_0$ is the number of non-zero entries of $x$. In many applications, such as embedding bag-of-words and tf-idf representations of documents, the input points are indeed very sparse compared to the domain size $d$ (one non-zero entry in $x$ per word in the document, where $d$ the number of distinct words in the dictionary).

\paragraph{Large Sets with Few Dimensions.}
While it may seem that there is little room for improvement in the $\lg(1/\eps)$ gap between the upper and lower bounds known for the sparsity of $\varepsilon$-JL matrices, we identify a shortcoming in the lower bound of Nelson and Nguyen~\cite{NN13}. Concretely, the hard instance in their proof is the set $\{e_1,\dots,e_n\}$ of standard unit vectors. However, in many theoretical applications, the original dimension $d$ is significantly smaller than the size $n$ of the vector-set. In these scenarios, this hard instance does not exist, in which case the lower bound degenerates to $s=\Omega(\eps^{-1} \lg d/\lg(m/\lg d))$. Yet, the upper bound analysis by Kane and Nelson is incapable of exploiting that $d \ll n$ and remains $O(\eps^{-1} \lg n)$.

In addition to broadening our theoretical understanding of sparse dimensionality reduction, we also find that $d \ll n$ is a natural practical setting, also when combined with sparse input points. Consider, for instance, the Sentiment140 data set consisting of 1.6M tweets~\cite{GBH09}. Using a bag-of-words or tf-idf representation of the tweets, where all words occuring less than $10$ times in the $1.6M$ tweets have been stripped, results in a data set with $n=1.6 \cdot 10^6$, $d=37,129$ and an average of $12$ words/non-zeros per tweet. These vectors are thus extremely sparse and have $d$ a factor $43$ less than $n$. Similarly, for the Kosarak data set~\cite{BKT18} consisting of an anonymized click-stream from a Hungarian online news portal, there are $n=900,002$ transactions, each consisting of several items. It has a total of $d=41,270$ distinct items and each transaction consists of an average of $8.1$ items. Here we also have an $n$ that is a factor $22$ larger than $d$ and very sparse input points. In general, when considering bag-of-words and tf-idf, one would assume that there is a fixed dictionary size $d$, while the number of data points $n$ may be arbitrarily large, which further motivates distinguishing between $n$ and $d$ in the sparsity bounds.

One may thus hope to give upper bounds with a dependency on $\lg d$ rather than $\lg n$. This is precisely the message of our work. 

\subsection{Main Results}

Our first main result is an improved analysis of the random sparse embedding by Kane and Nelson~\cite{KN14} reducing the $O(\eps^{-1} \lg n)$ upper bound on the sparsity in the case $d\ll n$. Formally we show the following.
\begin{theorem}
\label{th:upperBoundMain}
Let $0 < \eps < \eps_0$ for a constant $\eps_0$.  There is a distribution over $s$-sparse matrices in $\mathbb{R}^{m\times d}$ with $m=O(\eps^{-2} \lg n)$ and
\[
  s = O\left(\frac{1}{\eps} \cdot \left(\frac{\lg n}{\lg(1/\varepsilon)} + \lg^{2/3} n \lg^{1/3} d \right)\right)\;,
\]
such that for any set of $n$ vectors $X \subset \R^d$, it holds with probability at least $1-O(1/d)$ that a sampled matrix is an $\eps$-JL matrix for $X$.
\end{theorem}
While the first term may resemble the lower bound presented by Nelson and Nguyen~\cite{NN13}, their lower bound did not apply when the size of $X$ is significantly larger than the dimension $d$, and thus cannot consist of just the standard basis for $\mathbb{R}^d$. 

Our second result complements the upper bound in Theorem~\ref{th:upperBoundMain} with a tight lower bound on the sparsity of $\varepsilon$-JL matrices. We show that if $m$ is sufficiently smaller than $d$, then every $\varepsilon$-JL matrix embedding $d$-dimensional vectors in $\mathbb{R}^m$ must have relatively dense columns. Formally we show the following.
\begin{theorem}
  \label{thm:lower}
  Let $0 < \eps < 1/4$, and let $m$ be such that $m = \Omega(\eps^{-2} \lg n)$ and $m \leq (\eps d/\lg n)^{1-o(1)}$. Then there is a set of $n$ vectors $X \subset \R^d$ such that any $\eps$-JL matrix $A$ embedding $X$ into $\R^m$, must have a column with sparsity $s$ satisfying 
	\[
	s = \Omega\left(\frac{\lg n}{\varepsilon \lg(m/\lg n)}\right) \;.
	\]
For optimal $m = \Theta(\eps^{-2} \lg n)$, this simplifies to $s = \Omega(\eps^{-1} \lg n/\lg(1/\eps))$.
\end{theorem}
Recall the comparable lower bound in \cite{NN13} was specifically for the case $n=d$.
Combined with the refined upper bound, we now have a completely tight understanding of sparse dimensionality reduction when $\lg n \geq \lg d \cdot \lg^3(1/\eps)$. 
While arguably being small asymptotic improvements, these are the first improvements in a decade and demonstrate that the dimension of the input data may be exploited to speed up embeddings.

\paragraph{Subspace Embeddings.}
Given a $k$-dimensional subspace $V\subset \R^d$, an {\it $\eps$-subspace embedding} \cite{Sarlos06} is a matrix $A\in\R^{m\times d}$ satisfying that for all $x\in V$, $\|Ax\|_2^2 \in (1\pm \eps)\|x\|_2^2$. It is known that there exists a subset $V'\subset V$ of size $O(1)^k$ such that if $A$ preserves the $\ell_2$ norm of every vector in $V'$ up to $(1+\eps/2)$, then $A$ is an $\eps$-subspace embedding \cite{AroraHK06}. The JL lemma thus implies that one can take $m = O(k/\eps^2)$, and in fact this is optimal in the case that $A$ is drawn from a fixed distribution over $\R^{m\times d}$ that is independent of $V$ \cite{NelsonN14} (a so-called {\it oblivious subspace embedding (OSE)}). OSE's can be used to speed up algorithms for approximate regression, low rank approximation, and a large number of other problems in numerical linear algebra; see the monograph by Woodruff \cite{Woodruff14}.

As a simple example, consider the problem of approximate linear regression in which one wants to find a $\tilde\beta$ which approximately minimizes $\|X\beta - y\|_2^2$ for some given $X\in\R^{n\times d}$. This problem can be solved exactly in $O(nd^2)$ time by writing the Singular Value Decomposition $X = U\Sigma V^\top$ then setting $\beta_{\text{LS}} :=  V \Sigma^{-1}U^\top y$. Then $X \beta_{\text{LS}} = UU^\top y$ is the projection of $y$ onto the column space of $X$, which minimizes the error. The {\it sketch-and-solve} paradigm \cite{Sarlos06}, in one analysis, suggests taking $A$ to be a subspace embedding for $\mathop{span}\{y, \text{cols}(X)\}$ (which has dimension at most $d+1$) then setting $\tilde \beta$ to be the minimizer of $\|AX\beta - Ay\|_2^2$. Note $AX$ is now a much smaller matrix, so one can compute $\tilde\beta$ more quickly. However, we also need $A$ to either be sparse or structured, so that $AX$ can be computed quickly. Otherwise, if $A$ is an arbitrary unstructured matrix, computing $AX$ would take more time than computing $\beta_{\text{LS}}$ exactly!

Note that if each column of $A$ has $s$ nonzero entries, then $AX$ can be computed in time $O(s\|X\|_0)$, where $\|X\|_0$ is the number of nonzero entries in $X$. Simply using the SparseJL transform \cite{KN14} would lead to $m = O(k/\eps^2)$, $s = O(k/\eps)$. Clarkson and Woodruff \cite{ClarksonW13} showed that $m = O(k^2/\eps^2)$, $s = O(1)$ is achievable, which for OSE's is optimal \cite{NelsonN14,LiL22}. What though if we do not want to increase $m$ at all beyond the optimal bound of $O(k/\eps^2)$? What is the best sparsity $s$ achievable without sacrificing the asymptotic quality of dimensionality reduction? Nelson and Nguyen showed $m = O((k/\eps^2)\cdot \mathop{poly}(\eps^{-1}\log k))$ is achievable with $s = \mathop{poly}(\log(k/\eps))/\eps$ \cite{NelsonN13}, and conjectured that $s = O((\log k)/\eps)$ suffices with $m = O(k/\eps^2)$. Cohen provided an improved bound, showing $m = O((k\log k)/\eps^2)$, $s = O((\log k)/\eps)$ suffices \cite{Cohen16}, which remains the best known bound today. In particular, for $m = O(k/\eps^2)$, despite the conjecture of \cite{NelsonN13}, no sparsity bound better than $s = O(k/\eps)$ is known, which follows from black box application of SparseJL. In this work, we provide the first proof that keeps $m = O(k/\eps^2)$ while showing a sparsity bound that is $o(k/\eps)$. Specifically, we achieve $s = O(k/(\eps\log(1/\eps)) + \sqrt[3]{k^2\log k}/\eps)$. 
Formally we show the following.
 \begin{theorem}
 \label{th:subspacesMain}
 Let $0 < \eps < 1$. There is a distribution over $s$-sparse matrices in $\mathbb{R}^{m\times d}$ with $m=O(\eps^{-2} k)$ and
 \[
   s = O\left(\frac{1}{\eps} \cdot \left(\frac{k}{\lg(1/\varepsilon)} + k^{2/3} \lg^{1/3}k\right)\right)\;,
 \]
  such that for any $k$-dimensional subspace $V \subseteq \R^d$, it holds with probability at least $1-2^{-k^{2/3}}$ that a sampled matrix is an $\eps$-JL matrix for $V$.
 \end{theorem}
While this is far from the conjectured optimal bound of $O((\log k)/\eps)$, it provides the first analysis that maintains optimal $m$ while providing sparsity $s$ strictly better than applying SparseJL as a black box.



\section{Technical Overview}
In this section, we present the central ideas employed in our new contributions. We first survey our improved upper bound analysis, then the main ideas in our lower bound, and finally the new subspace embedding results. For ease of notation, we henceforth write $\|x\|$ to denote $\|x\|_2$.

\paragraph{Sparser Dimensionality Reduction.}
One method for achieving Sparse JL matrices presented by Kane and Nelson \cite{KN14} is based on the \textsf{CountSketch} algorithm \cite{CharikarCF04}. An embedding matrix $A$ is sampled by partitioning the $m$ rows into $s$ groups of $m/s$ entries each. In every column of $A$ a uniform random entry in each group is sampled and set uniformly to either $1/\sqrt{s}$ or $-1/\sqrt{s}$. All other entries are set to $0$. Kane and Nelson then showed that if $s = \Omega(\eps^{-1} \lg(1/\delta))$ then for every unit vector $x$, it holds that $\|Ax\|^2 \in 1 \pm \eps$ with probability at least $1-\delta$. Setting $\delta = n^{-3}$, using linearity of $A$ and a union bound over $z = (y-x)/\|y-x\|$ for all $x,y$ in an input set of points/vectors $X$ completes their proof. Hereafter we focus on showing that $A$ preserves the norm of every vector in a set $X$ of $n^2$ unit vectors with good probability.
Kane and Nelson also included a short argument showing that their analysis is tight for distances between the standard unit vectors $e_1,\dots,e_d$. 

However, our key observation is that, if $d \ll n$, then a naive union bound over all $n^2$ pairs of vectors in $X$ may be too loose. Concretely, there are much fewer than $n^2$ vectors that are of this worst case form. In particular, when $d \ll n$, then most vectors in a set $X$ of cardinality $n$ must have many entries that are small in magnitude. It is already known from work on Feature Hashing~\cite{WKD+09,DKT17,FKL18,J19} and the FastJL transform~\cite{AC09, FML22} that vectors $x$ with a small $\|x\|_\infty$ to $\|x\|$ ratio are easier to embed than worst case vectors. For instance, for optimal $m = \Theta(\eps^{-2} \lg n)$, Jagadeesan~\cite{J19} showed that as long as $s = \Omega(\eps^{-1} \lg n/\lg(1/\eps))$  and the ratio $\nu = \|x\|_\infty/\|x\|$ satisfies $\nu \leq \sqrt{\eps s /\lg n}$,
then SparseJL preserves the norm of $x$ to within $1 \pm \eps$ with probability at least $1-1/n^3$.

In order to exploit a small dimension $d$, we split every vector $x \in X$ into two support-disjoint vectors, referred to as a \emph{head} and a \emph{tail}, where the head contains the top $\ell$ entries of $x$ and the tail contains the remaining entries. That is, we write $x = x_{head} + x_{tail}$. Then \[\|Ax\|^2 = \|Ax_{head}\|^2 + \|Ax_{tail}\|^2 + 2\langle Ax_{head}, Ax_{tail} \rangle \;.\] 
We now treat these three terms separately. Showing that with high probability, $\|Ax_{head}\|^2 \in (1 \pm \varepsilon)\|x_{head}\|^2$, $\|Ax_{tail}\|^2 \in (1 \pm \varepsilon)\|x_{tail}\|^2$ and $|\langle Ax_{head}, Ax_{tail} \rangle| \le \varepsilon$ (since $\langle x_{head}, x_{tail}\rangle =0$). The technical crux lies in bounding the cross terms.

In order to bound the heads, the main observation is that there are about $\binom{d}{\ell} \leq d^\ell$ choices for the positions of the heads. Once the positions have been chosen, we further approximate the heads by an $\eps$-net of cardinality $(1/\eps)^{O(\ell)}$.
Since $d \geq m = \Omega(\eps^{-2} \lg n)$, the total number of heads we need to consider is $d^\ell (1/\eps)^{O(\ell)} = d^{O(\ell)}$. Using the analysis by Kane and Nelson with $\delta = d^{-O(\ell)}$ shows that it suffices with $s = \Omega(\eps^{-1} \ell \lg d)$ to get the required bound with high probability.

As for the tails, there are at most $n^2$ distinct tails and they have $\|x_{tail}\|_\infty \leq 1/\sqrt{\ell} \le \|x_{tail}\|_2 /\sqrt{\ell} $. We can thus use the result by Jagadeesan to show that $\|Ax_{tail}\|^2 \in (1 \pm \eps)\|x_{tail}\|^2$ whenever $s$ satisfies both $s = \Omega(\eps^{-1} \lg n/\lg(1/\eps))$ and $(1/\sqrt{\ell}) \leq \sqrt{\eps s /\lg n}$, which is implied by $s = \Omega(\eps^{-1} \lg(n)/\ell)$.

The main challenge lies in bounding the cross terms showing $|\langle Ax_{head}, Ax_{tail} \rangle|  \le \eps$. 
Previous results, and specifically the aforementioned results by Kane and Nelson \cite{KN14} and Jagadeesan \cite{J19} cannot be employed, as on one hand the number of pairs is very large, and more specifically depends polynomially on $n$, and on the other hand the $\ell_\infty/\ell_2$ ratio of the corresponding vectors cannot be upper bounded as the heads have heavy entries. In order to bound the cross terms we present new concentration bounds on the CountSketch-based construction by Kane and Nelson. We first show that for optimal dimension $m = O(\varepsilon^{-2}\lg n)$, for sparsity $s \le \varepsilon m$ and $\ell \le \varepsilon^{-1/2}$ we get that with high probability for every $x \in X$, there are only few rows in $A$ where more than $5$ non-zero entries coincide with the support of $x_{head}$. In turn, this means that most entries of $Ax_{head}$ are not too large, and specifically do not exceed $\sqrt{5/s}$. We then turn to analyze the probability that for some $x \in X$, 
\[
\langle Ax_{head}, Ax_{tail} \rangle = \sum_{i \in [m]}(Ax_{head})_i(Ax_{tail})_i = \sum_{i \in [m]}\left((Ax_{head})_i\sum_{j \in \supp(x_{tail})}a_{ij}x_j\right)
\]
is at most $\varepsilon$. To this end, we partition the sum into two sums, where the first sum handles terms $(i,j)$ where $(Ax_{head})_i$ and $x_j$ are large and the second sum handles the remaining terms. Here we exploit that $(Ax_{head})_i$ is small for most $i$ as just argued. Furthermore, since $x$ is unit length, there are also few choices of $j$ where $x_j$ is large. The first sum can thus be handled by exploiting that there are few terms in the sum, and the second sum has strong concentration since the terms are small.

\paragraph{Stronger Lower Bound.}
To improve over the lower bound given by Nelson and Nguyen~\cite{NN13}, we first need to define a harder input instance. Concretely, they used the standard unit vectors $e_1,\dots,e_n$, which as argued earlier, only is a valid input for $d \geq n$.

Our hard instance $X$ instead consists of all vectors of the form $v_S=\sum_{i \in S} e_i/\sqrt{|S|}$ for subsets $S \subseteq [d]$ of cardinality $\lg n/\lg d$, all the standard unit vectors $e_1,\dots,e_d$, as well as the origin $0$.

Now consider an $m \times d$ embedding matrix $A$, such that each column of $A$ has at most $s$ non-zeros, and $A$ is an $\varepsilon$-JL matrix for $X$. Since $e_1,\dots,e_d$ and $0$ are in the input, it must be the case that each column $a_j$ of $A$ has norm in $1 \pm \eps \leq 2$. Now assume for simplicity that all the columns of $A$ had precisely $s$ non-zero entries and those took values $\{-1/\sqrt{s},1/\sqrt{s}\}$. For a subset $T \subseteq [m]$ of $t$ entries and a list of $t$ signs $\sigma=(\sigma_1,\dots,\sigma_t)$, we say that $a_j$ has the signature $(T,\sigma)$ if $a_j$ is non-zero in every coordinate corresponding to $T$ and its coordinates inside $T$ have the signs $\sigma$. Any column would then have $\binom{s}{t}$ distinct signatures. Since there are $\binom{m}{t}2^t$ signatures and $d$ columns, it follows by averaging that there must be a signature shared by at least $d \binom{s}{t}/\binom{m}{t}2^t \approx d (s/m)^t$ columns. We set $t$ roughly as $c \lg d/\lg(m/s)$ for a small constant $c>0$, resulting in at least $\textrm{poly}(d)$ columns sharing the same signature.

We now fix such a signature and let $S$ be the subset of columns in $A$ with that signature. If $|S| = \textrm{poly}(d) \geq \eps^{-1} \lg n/\lg d$, then we can select $\eps^{-1}$ disjoint subsets $S_1,\dots,S_{\eps^{-1}}$ of $S$, each of cardinality $\lg n/\lg d$. For each such subset $S_i$, we know that the vector $v_{S_i}$ is in $X$. Now inside the coordinates in $T$, all columns in $S_i$ are non-zero and have the same sign. Hence the entries of $Av_{S_i}$ inside $T$ are $\sqrt{\lg n/(s \lg d)}$ in magnitude as the columns add up. Moreover, the entries inside $T$ also have the same signs across distinct $Av_{S_i}$ and $Av_{S_j}$.

If we now delete the entries in $T$ from all such $Av_{P_i}$, we are left with vectors whose norm is no more than $1 + \eps < 2$. Moreover, since $v_{S_i}$ and $v_{S_j}$ have disjoint supports, they were orthogonal before embedding and thus to preserve their distance, the inner products of $Av_{S_i}$ and $Av_{S_j}$ must be $O(\eps)$. Deleting the entries in $T$ reduces these inner products by $|T| \lg n/(s \lg d) = t \lg n/(s \lg d) \approx \lg n/(s \lg(m/s))$. If we call the resulting vectors $\tilde{A}v_{S_i}$, then it must hold that $0 \leq \|\sum_{i=1}^{\eps^{-1}} \tilde{A}v_{S_i} \|^2 = \sum_{i=1}^{\eps^{-1}} \|\tilde{A}v_{S_i}\|^2 + \sum_i \sum_{j \neq i} \langle \tilde{A}v_{S_i} , \tilde{A}v_{S_j} \rangle \leq 2\eps^{-1} + \eps^{-1} (\eps^{-1} -1) (O(\eps) - \lg n/(s \lg(m/s)))$. Multiplying by $\eps$ and solving for $s$ gives $s = \Omega(\eps^{-1} \lg n/\lg(m/s))$. Since $m = \Omega(\eps^{-2} \lg n)$, this is equivalent to $s = \Omega(\eps^{-1} \lg n/\lg(m/\lg n))$.

To deal with columns of $A$ that are not of the form $\{-1/\sqrt{s},0,1/\sqrt{s}\}$ we redefine signatures to be subsets of coordinates where $a_j$ has large norm restricted to those coordinates. Also, instead of the signs $\sigma$, we instead build a $1/4$-net over the $T$ coordinates and let the closest net point be a substitute for the signs. 


Comparing our argument to that of Nelson and Nguyen~\cite{NN13}, the key difference lies in summing up multiple columns of $A$ that all share the same signature. To ensure this sum of columns corresponds to a vector in $X$, we add every sum of $\lg n/\lg d$ columns $v_S$ to the input.

\paragraph{Subspace Embeddings.}
For subspace embeddings, we note that the classic approach for showing a sparsity of $s = O(\eps^{-1} k)$ follows by constructing a $1/2$-net $\enet \subset V$ over the $k$-dimensional subspace $V$. One can then (roughly) show that if a linear embedding matrix $A$ preserves the norm of all net points, then it preserves the pairwise distance between all points in $V$. Since such a net has cardinality $2^{O(k)}$, the claimed sparsity follows from Kane and Nelson's $s = O(\eps^{-1} \lg(1/\delta))$ with $\delta = 2^{-O(k)}$.

A first attempt at improving over this would be to directly insert $n=2^{O(k)}$ into our improved sparse embedding from above. This would result in a sparsity of $s = O(\eps^{-1}(k/\lg(1/\eps) + k^{2/3} \lg^{1/3} d))$. The first term is fine, but the latter term depends on $d$, which would make the bound incomparable to previous results that only depend on $k$ and $\eps$. We thus take a closer look at the origin of the dependency on $d$.

Recall that for a set of $n$ vectors, such as the net $\enet$, we partition the vectors $w \in \enet$ into a head and a tail as $w = w_{head} + w_{tail}$ where $w_{head}$ contains the largest $\ell$ entries of $w$. We then observed that for an embedding matrix $A$, we have
\[
  \|Aw\|^2 = \|Aw_{head}\|^2 + \|Aw_{tail}\|^2 + 2 \langle Aw_{head}, Aw_{tail} \rangle.
\]
We then showed that  $\|Aw_{head}\|^2 \in (1 \pm O(\eps))\|w_{head}\|^2$, $\|Aw_{tail}\|^2 \in (1 \pm O(\eps))\|w_{tail}\|^2$ and $|\langle Aw_{head}, Aw_{tail} \rangle| = O(\eps)$. For the second term, we exploited that $\|w_{tail}\|_\infty \leq 1/\sqrt{\ell}$ and then combined this with the result by Jagadeesan for embedding vectors with a small $\| \cdot \|_\infty$. The requirement on $s$ resulting from this term was $s = \Omega(\eps^{-1} \lg n/\lg(1/\eps)) = \Omega(\eps^{-1} k /\lg(1/\eps))$ as well as $s = \Omega(\eps^{-1} \lg(n)/\ell) = \Omega(\eps^{-1} k/\ell)$. Hence no dependencies on $d$ here. Similarly, for the cross terms, we got the requirement $s = \Omega(\eps^{-1} \lg(n)/\sqrt{\ell}) = \Omega(\eps^{-1} k/\sqrt{\ell})$. Thus the dependency on $d$ comes only from preserving the norms of the heads.

For the heads, we argued that there were $\binom{d}{\ell}$ choices for the positions of the heads and thereafter, we needed an $\eps$-net on the chosen $\ell$ positions. This resulted in $d^{O(\ell)}$ heads in the net and we then used Kane and Nelson's analysis yielding $s =O(\eps^{-1} \lg(1/\delta))$ with $\delta = d^{-O(\ell)}$. Thus we need a tighter bound on the number of heads to avoid the dependency on $d$.

The first idea is to change the definition of the head $w_{head}$ to be all entries $w_i$ of $w$ with $|w_i| \geq 1/\sqrt{\ell}$. This is a small but crucial change from the previous definition where the head contained the top $\ell$ entries. To distinguish the two, we instead denote the heavy entries by $w_{heavy}$ and the remaining entries by $w_{light} = w - w_{heavy}$.

Next, we argue that the positions of the at most $\ell$ entries in $w_{heavy}$, must be among a small set of coordinates:
\begin{lemma}
  \label{lem:subspacecover}
Let $V$ be a $k$-dimensional subspace of $\R^d$. For every $\ell \ge 1$, there is a set $S \subseteq [d]$ of coordinates with $|S| \leq k \ell$ such that for every unit vector $v \in V$, all coordinates $i \in [d] \setminus S$ satisfy $|v_i| < 1/\sqrt{\ell}$.
\end{lemma}
Lemma~\ref{lem:subspacecover} states that the positions of the non-zeros in all $w_{heavy}$ must be contained in a small set $S$ of cardinality only $|S| = k \ell$. Thus there are only $\binom{|S|}{\ell} =2^{O(\ell \lg k)}$ possible positions of the non-zeros in $w_{heavy}$. Next, we also argue that once the positions of the heavy entries have been determined, it suffices with $1/2$-net on the chosen positions. Hence we reduce the number of $w_{heavy}$ to just $2^{O(\ell \lg k)}$ and have removed the dependency on $d$. Using Kane and Nelson now gives us that we need $s = \Omega(\eps^{-1} \ell \lg k)$. Balancing this with $s = \Omega(\eps^{-1} k/\sqrt{\ell})$ gives $\ell = (k /\lg k)^{2/3}$. The final bound thus becomes $s = O(\eps^{-1}(k/\lg(1/\eps) + k^{2/3}\lg^{1/3} k))$ as claimed.

\section{Sparsity Upper Bound}
\label{sec:upper}
In this section we prove Theorem~\ref{th:upperBoundMain}. Let $d$ be an integer, let $\varepsilon \in (0,1)$ and let $X \subseteq \mathbb{R}^d$ be some finite set of $n$ vectors. Let $m = O(\varepsilon^{-2}\lg n)$ and let $s = O(\varepsilon^{-1}\lg n/\lg(1/\varepsilon) + \varepsilon^{-1}\lg^{2/3}n\lg^{1/3}d)$.
We will show that if $A$ is sampled as in Kane and Nelson \cite{KN14}, then $A$ is an $\varepsilon$-JL matrix for $X$ with probability at least $1 - O(1/d)$.

For simplicity, we will actually only show that it is an $O(\eps)$-JL matrix. A simple rescaling of $\eps$ by a constant factor implies the result.



As $A$ is a linear transformation, and $n$ appears in all terms inside a logarithm, it is enough to show the following claim (by replacing $X$ with $X'$ containing $x_{i,j} = (x_i - x_j)/\|x_i-x_j\|$ for all $x_i,x_j \in X$).
\begin{claim} \label{c:nUnitVectors}
Assume $A$ is sampled as in Kane and Nelson \cite{KN14} with $m = O(\varepsilon^{-2}\lg n)$ and $s = O(\varepsilon^{-1}\lg n/\lg(1/\varepsilon) + \varepsilon^{-1}\lg^{2/3}n\lg^{1/3}d)$, then for every set $X \subseteq \mathbb{R}^d$ of $n$ unit vectors, it holds that with probability at least $1 - O(1/d)$ for all $x \in X$ that $\|Ax\|^2 \in (1 \pm O(\varepsilon))$.
\end{claim}
For the rest of the section we therefore prove Claim~\ref{c:nUnitVectors}, and we start by introducing the following notation.

\begin{notation}
Let $x \in \mathbb{R}^{d}$, and let $\ell \in [d]$. Denote by $x_{head(\ell)}$ the vector obtained from $x$ where all but the top $\ell$ entries are zeroed out. Denote $x_{tail(\ell)} = x - x_{head(\ell)}$.
\end{notation}

Let $\ell = \left\lceil \min\left\{\eps^{-1/2}, \left(\frac{\lg n}{\lg d}\right)^{2/3} \right\} \right\rceil$ be an integer. For every $T \in \binom{[d]}{\ell}$, let $\Y_T$ be the set of all vectors $y \in \mathbb{R}^d$ such that $\|y\| \le 1$ and $\supp(y) \subseteq T$. Let $\Y = \bigcup_{T \in \binom{[d]}{\ell}}\Y_T$. Note that for every $i \in [d]$, $e_i \in \Y$.

Fix some set $X \subseteq \mathbb{R}^d$ of $n$ unit vectors.
Define ${\cal E}_1$ to be the set of all matrices $A \in \mathbb{R}^{m\times d}$ such that for all $x \in \Y$, $\|Ax\|^2 \in (1 \pm \varepsilon)\|x\|^2$.
Define ${\cal E}_2$ to be the set of all matrices $A \in \mathbb{R}^{m\times d}$ such that for all $x \in X$, $\|Ax_{tail(\ell)}\|^2 \in \|x_{tail(\ell)}\|^2 \pm \varepsilon$.
Define ${\cal E}_3$ to be the set of all matrices $A \in \mathbb{R}^{m\times d}$ such that for all $x \in X$, either $\|Ax_{head(\ell)}\|^2 > 2$ or $\left|\dotp{Ax_{head(\ell)}}{Ax_{tail(\ell)}}\right| < \varepsilon$.

\begin{claim}
Assume $A \in {\cal E}_1 \cap {\cal E}_2 \cap {\cal E}_3$. Then for every $x \in X$, $\|Ax\|^2 \in (1 \pm O(\varepsilon))$.
\end{claim}

\begin{proof}
Let $x \in X$, then $\|Ax\|^2 = \|Ax_{head(\ell)}\|^2 + \|Ax_{tail(\ell)}\|^2 + 2 \dotp{Ax_{head(\ell)}}{Ax_{tail(\ell)}}$. If $A \in {\cal E}_1$, then $\|Ax_{head(\ell)}\|^2 \in (1 \pm \varepsilon)\|x_{head(\ell)}\|^2$. Specifically $\|Ax_{head(\ell)}\|^2  < 2$ and thus since we also have $A \in {\cal E}_3$, it must be the case that $\left|\dotp{Ax_{head(\ell)}}{Ax_{tail(\ell)}}\right| < \varepsilon$.
Therefore 
\[
\|Ax\|^2 \le (1+\varepsilon)\|x_{head(\ell)}\|^2 + \|x_{tail(\ell)}\|^2 + \varepsilon + 2 \varepsilon \le (1+O(\varepsilon))
\]
Similarly
\[
\|Ax\|^2 \ge (1-\varepsilon)\|x_{head(\ell)}\|^2 + \|x_{tail(\ell)}\|^2 -\varepsilon - 2 \varepsilon \ge (1-O(\varepsilon))
\]
\end{proof}

It remains to show that $\Pr[A \in {\cal E}_1 \cap {\cal E} \cap {\cal E}_3] \ge 1 - O(1/d)$. This is implied by the next three claims, bounding the probability of each of the events separately.

\begin{claim}
$\Pr[A \in {\cal E}_1] \ge 1- d^{-1}$.
\end{claim}

\begin{proof}
Denote $\delta = 2^{-\Omega(\sqrt[3]{\lg^2 n \lg d})}$. As $n \ge d$ we get that $m \ge \Omega(\varepsilon^{-2}\lg (1/\delta))$ and $s \ge \Omega(\varepsilon^{-1}\lg (1/\delta))$. Following Kane and Nelson \cite{KN14}, for every unit vector $x \in \mathbb{R}^d$ we have that $\Pr\left[\|Ax\|^2\in (1\pm \varepsilon)\right] \ge 1 - \delta$.
Denote by $\hat{\Y}$ the set of all vectors $y \in \Y$ such that for every $i \in T$, $d^3y_i$ is an integer. Then for every $y \in \hat{\Y}$, for every $i \in \supp(y)$, $d^3y_i \in \{-d^3,-d^3+1,\ldots,d^3\}$, and therefore 
\[
|\hat{\Y}| \le \binom{d}{\ell}(2d^3)^{\ell} \le (2d^4)^{\ell} \le 2^{O(\ell \lg d)} = 2^{O(\sqrt[3]{\lg^2 n \lg d})} \le \frac{1}{\sqrt{\delta}} \;.
\]

Therefore with probability at least $1 - \sqrt{\delta} \ge 1 - d^{-1}$, we get that for all $y \in \hat{\Y}$, $\|Ay\|^2\in (1\pm \varepsilon)$.

Assume therefore that for all $y \in \hat{\Y}$, $\|Ay\|^2\in (1\pm \varepsilon)$. Let $x \in \Y$, and let $y \in \hat{\Y}$ be the closest vector in $\hat{\Y}$ to $x$. Then $\|x- y\|^2 \le \ell/d^3$. Since $\|A\|_F^2 = d$ we get that
\[
\|A(x-y)\| \le \|A\|_F \|x- y\| \le \sqrt{\frac{sd\ell}{d^3}} \le O\left(\sqrt{\frac{(\lg n)^{5/3}}{\varepsilon d^2}}\right) = O(\varepsilon) \;,
\]
where the last inequality is due to the fact that $d \ge \frac{\lg n}{\varepsilon^2}$. Therefore
\[
\|Ax\| \le \|Ay\| + \|A(x-y)\| \le 1 + \varepsilon + O(\varepsilon) \le 1 + O(\varepsilon), 
\]
and similarly $\|Ax\| \ge 1-O(\varepsilon)$.
\end{proof}

For the second claim, we make use of the following result by Jagadeesan~\cite{J19}.
\begin{theorem}[Jagadeesan~\cite{J19}]
 \label{thm:meenah}
 For any $0 < \delta < 1$ and $0 < \eps < \eps_0$ for some constant $\eps_0$, assume $A$ is sampled as in Kane and Nelson \cite{KN14} with $m \geq \Theta(\eps^{-2} \lg(1/\delta))$ and $m \geq s \exp(\max\{1,\ln(1/\delta)/(\eps s)\})$, then for any vector $v$ with
 \[
   \frac{\|v\|_\infty}{\|v\|} = O\left(\sqrt{\frac{\eps s \ln(m \eps^2/\ln(1/\delta))}{\ln(1/\delta)}}\right)
 \]
 we have $\Pr[\|Av\|^2 \in (1 \pm \eps)\|v\|^2] \geq 1-\delta$.
\end{theorem}

Using Theorem~\ref{thm:meenah}, we now turn to bound the probability of ${\cal E}_2$.

\begin{claim} \label{c:tails}
$\Pr[A \in {\cal E}_2] \ge 1-n^{-1}$.
\end{claim}

\begin{proof}
Fix $x \in X$. Denote $v = x_{tail(\ell)}$, and let $\hat{\varepsilon} = \max\{\varepsilon, \frac{\lg n}{s}\left(\frac{\|v\|_\infty}{\|v\|}\right)^2 \}$. We wish to apply Theorem~\ref{thm:meenah} and thus start by verifying that our choice of parameters satisfy the constraints in the theorem. Applying the right constants, we have that $m \ge \Theta(\varepsilon^{-2}\log n) \ge \Theta(\hat{\varepsilon}^{-2}\log n)$. Furthermore
\begin{equation*}
\begin{split}
s e^{\Theta(\max\{1 , (\hat{\varepsilon} s)^{-1}\lg n\})}  &\le s e^{\Theta(\max\{1 , (\varepsilon s)^{-1}\lg n\})} \le s e^{\max\{\Theta(1) , -\lg \varepsilon\}} 
= s \cdot \max\{e^{\Theta(1)}, \frac{1}{\varepsilon}\} \\ &= O\left(\frac{1}{\varepsilon}\left(\sqrt[3]{\lg^2 n \lg d} + \lg n/\lg(1/\eps)\right)
\right)\max\{e^{\Theta(1)}, \frac{1}{\varepsilon}\} \le m.
\end{split}
\end{equation*}

Finally note that
\[
\sqrt{\hat{\varepsilon}s \frac{\lg \frac{m \hat{\varepsilon}^2}{\lg n}}{\lg n}} \ge \sqrt{\frac{\lg n}{s}\left(\frac{\|v\|_\infty}{\|v\|}\right)^2\cdot s \frac{1}{\lg n}} \ge \frac{\|v\|_\infty}{\|v\|}
\]
Therefore, Theorem~\ref{thm:meenah} gives us that with probability $\ge 1 - \frac{1}{n^{2}}$ we have $\|Av\|^2 \in (1\pm \hat{\varepsilon})\|v\|^2$. That is $\|Av\|^2 \in \|v\|^2 \pm \hat{\varepsilon}\|v\|^2 = \|v\|^2 \pm \max\{\varepsilon \|v\|^2, \frac{\lg n}{s}\|v\|_\infty^2 \}$. Note first that $\varepsilon \|v\|^2 < \varepsilon$. Next, as $v = x_{tail(\ell)}$, and $\|x\| = 1$ we have that $\|v\|_\infty \le \frac{1}{\sqrt{\ell}}$, and therefore $\frac{\lg n}{s}\|v\|_\infty^2 \le \frac{\lg n}{s\ell} = O\left(\frac{1}{s} \cdot \max\{\lg^{1/3} n \lg^{2/3} d\;,\; \eps^{1/2} \lg n \}\right) = O(\eps)$. We conclude that with probability $\ge 1 - \frac{1}{n^{2}}$ we have that $\|Ax_{tail(\ell)}\|^2 \in \|x_{tail(\ell)}\|^2\pm O(\varepsilon)$. Applying a union bound we get that $\Pr[A \in {\cal E}_2] \ge 1 - \frac{1}{n}$.
\end{proof}

The following claim, that essentially shows that with high probability over the choice of $A$, the cross terms are small constitute the technical crux of the upper bound result, and its proof requires a much more careful examination of the construction by Kane and Nelson~\cite{KN14}.

\begin{claim}
  \label{c:cross}
$\Pr[A \in {\cal E}_3] \ge 1 - n^{-1}$.
\end{claim}

%
Recall that for a choice of $m$ and $s$, the construction works by grouping the rows of $A$ into $s$ blocks of $m/s$ consecutive rows each, $[1,m/s], [m/s+1,2m/s]$ and so on. For every column, a uniform random entry in each block is chosen together with an independent uniform sign $\sigma$. That entry is then set to $\sigma/\sqrt{s}$. Each column thus has exactly one non-zero per block of rows.

The rest of this section is devoted to the proof of Claim~\ref{c:cross}.
We start by proving that $x_{head(\ell)}$ often has a desirable property. Concretely, we define the following
\begin{definition}
A set $J \in \binom{[d]}{\ell}$ of $\ell$ columns of $A$
is \emph{well-behaved} if there are no more than $6\lg n/\lg(1/\eps)$ rows $i \in [m]$ such that $|\{j \in J : a_{ij} \neq 0\}| \geq 6$.
\end{definition}
\begin{claim}
  \label{c:oftenwell}
Let $A$ be sampled as in Kane and Nelson \cite{KN14} with $m = O(\eps^{-2}\lg n)$ and $s \leq \eps m$. Then for every 
set $J \in \binom{[d]}{\ell}$ of $\ell$ columns 
of $A$, it holds with probability at least $1-n^{-3}$ that $J$ is well-behaved.
\end{claim}
\begin{proof}
  Let $J \in \binom{[d]}{\ell}$, and denote $\beta = 6\lg n/\lg(1/\eps)$. For every subset $I =\{i_1,\ldots,i_\beta\}\in \binom{[m]}{\beta}$ of $\beta$ rows and every sequence $V_1,\dots,V_\beta \in \binom{J}{6}$ of subsets of $J$ of size $6$ each, define the event ${\cal E}_{I,V_1,\ldots,V_\beta}$ to be the set of all matrices $A$ such that for every $i \in I$, for every $j \in V_i$, $a_{i,j} \ne 0$.

Note first that if the entries $\{a_{i,j}\}_{i \in I, j \in V_i}$ are not independent, then there must be two such entries in the same column and same subset of $m/s$ rows. In this case, $\Pr[{\cal E}_{I,V_1,\ldots,V_\beta}] = 0$ as at most one of them may be non-zero. Otherwise, all $6\beta$ entries of $A$ considered in ${\cal E}_{I,V_1,\ldots,V_\beta}$ are independent and therefore $\Pr[{\cal E}_{I,V_1,\ldots,V_\beta}] \le \left(\frac{s}{m}\right)^{6\beta}$. As $s \le \varepsilon m$ and $\ell \le \varepsilon^{-1/2}$, and by applying a union bound we get that
\[
\Pr[\text{$J$ is not well-behaved}] \le \sum_{I \in \binom{[m]}{\beta}}\sum_{V_1,\dots,V_\beta \in \binom{J}{6}}\Pr[{\cal E}_{I,V_1,\ldots,V_\beta}] \le \left(\frac{me}{\beta}\right)^\beta\cdot \ell^{6\beta} \varepsilon^{6\beta} \le (O(1)\varepsilon^{-2}\lg(1/\varepsilon))^\beta \cdot \varepsilon^{3\beta} \;.
\]
For $\varepsilon$ smaller than some constant, this is at most $\varepsilon^{\beta/2} \le n^{-3}$.
%
%
\end{proof}

Next we show that if the the support of $x_{head(\ell)}$ is a well-behaved subset of columns, then $Ax_{head(\ell)}$ has few "large" entries (note that $|\supp(x_{head(\ell)})|\le \ell$).
\begin{claim}
  \label{lem:nice}
Let $x \in X$ and assume the support of $x_{head(\ell)}$ is well-behaved. Then $Ax_{head(\ell)}$ has at most $6\lg n/\lg(1/\eps)$ entries that exceed $\sqrt{5/s}$. Furthermore, we have $\|Ax_{head(\ell)}\|_\infty \leq \sqrt{\ell/s}$.
\end{claim}
\begin{proof}
Let $I \subseteq [m]$ be the set of rows $i \in [m]$ for which $|\{j \in \supp(x_{head(\ell)}) : a_{ij}\ne 0\}| \ge 6$.
Consider an $i \in [m] \setminus I$. The number of columns $j \in [d]$ such that $a_{ij}\ne 0$ is at most $5$ and for each of these we have $|a_{ij}| \le 1/\sqrt{s}$. Therefore since $\|x_{head(\ell)}\| \le 1$ we get that $|(Ax_{head(\ell)})_i| \le \sqrt{5/s}$.
Since $\supp(x_{head(\ell)})$ is well-behaved, then $|I| \le 6\lg n/\lg(1/\eps)$, and the claim follows. 
The bound $\|Ax_{head(\ell)}\|_\infty \leq \sqrt{\ell/s}$ follows simply from the support of $x_{head(\ell)}$ only having cardinality $\ell$ and $\|x_{head(\ell)}\|\leq 1$.
\end{proof}


For every $x \in X$, let ${\cal E}_{3,x}$ be the set of matrices $A$ where $\|Ax_{head(\ell)}\|^2 \geq 2$ or $|\langle Ax_{head(\ell)},Ax_{tail(\ell)} \rangle| < \eps$. Then ${\cal E}_3 = \cap_{x \in X} {\cal E}_{3,x}$.
Our goal is to show that $\Pr[{\cal E}_{3,x}] \geq 1-O(1/n^2)$ which by a union bound over all $x \in X$ completes the proof of Claim~\ref{c:cross}.
To this end, define ${\cal W}_x$ to be the set of all matrices $A$ for which the support of $x_{head(\ell)}$ is a well-behaved set of columns. Claim~\ref{c:oftenwell} implies that $\Pr[{\cal W}_x] \ge 1-n^{-3}$.
It is therefore enough to show that $\Pr[{\cal E}_{3,x} \mid {\cal W}_{x}]\ge 1-O(n^{-2})$, as $\Pr[{\cal E}_{3,x}] \ge \Pr[{\cal E}_{3,x} \mid {\cal W}_{x}]\Pr[{\cal W}_{x}]$.
The following lemma thus concludes the proof of Claim~\ref{c:cross}, and the rest of this section is devoted to its proof.

\begin{lemma}
\label{l:smallCrossCondWell}
$\Pr[{\cal E}_{3,x} \mid {\cal W}_{x}]\ge 1-O(n^{-2})$.
\end{lemma}

Since $\Pr[{\cal E}_{3,x} \mid {\cal W}_{x}\wedge\|Ax_{head(\ell)}\|^2 \ge 2]=1$ it is enough to bound $\Pr[{\cal E}_{3,x} \mid {\cal W}_{x}\wedge\|Ax_{head(\ell)}\|^2 < 2]$.
%
Note that by disjointness of the support of $x_{head(\ell)}$ and $x_{tail(\ell)}$, the vectors $Ax_{head(\ell)}$ and $Ax_{tail(\ell)}$ are independent. In fact, $Ax_{tail(\ell)}$ is completely independent of all columns of $A$ in the support of $x_{head(\ell)}$. 
%
We will therefore show that conditioned on ${\cal W}_x \wedge\|Ax_{head(\ell)}\|^2 < 2$, $|\langle Ax_{head(\ell)},Ax_{tail(\ell)} \rangle| = O(\eps)$ with probability at least $1-O(1/n^2)$ over the choice of the random columns in the support of $x_{tail(\ell)}$.
We can therefore condition on some outcome of $u = Ax_{head(\ell)}$ where $\supp(x_{head(\ell)})$ is also well-behaved.

For every $i \in [m]$ and $j \in \supp(x_{tail(\ell)})$ define $b_{ij}$ as the Bernoulli random variable taking the value $1$ if entry $(i,j)$ of $A$ is non-zero and $0$ otherwise. In addition, let $\sigma_{ij}$ denote uniform random and independent signs.
Then $\langle u, Ax_{tail(\ell)}\rangle = \sum_{i=1}^m u_i \sum_{j\in\supp(x_{tail(\ell)})} b_{ij} \sigma_{ij} x_j/\sqrt{s}$. To bound the sum, we split it into two sums, and bound the probabilities of each part being at most $O(\varepsilon)$. Denote 
\begin{equation*}
\begin{split}
R &= \{(i,j) \in [m]\times\supp(x_{tail(\ell)}) : |u_i| > \sqrt{5/s} \; and \; |x|_j > 1/(\sqrt{\ell}\lg^2(1/\varepsilon))\}\;, \; and \\
S &= ([m]\times\supp(x_{tail(\ell)})) \setminus R
\end{split}
\end{equation*}


%

%
\begin{claim}\label{c:crossLargeTerms} 
$\Pr\left[\left|\sum\limits_{(i,j)\in R}u_i \cdot b_{ij} \sigma_{ij} x_j/\sqrt{s}\right| \le O(\varepsilon) \right] \ge 1-n^{-2}$.
\end{claim}
\begin{proof}
Recall that $\|u\|_\infty \le \sqrt{\ell/s}$ and $\|x_{tail(\ell)}\|_\infty \leq 1/\sqrt{\ell}$.
Therefore 
\[
\left|\sum_{(i,j)\in R}u_i \cdot b_{ij} \sigma_{ij} x_j/\sqrt{s}\right|\le \frac{1}{\sqrt{s}}\sum_{(i,j)\in R}|u_i| \cdot b_{ij} |\sigma_{ij} x_j| \le \frac{1}{s}\sum_{(i,j)\in R}b_{ij} \;.
\]
To complete the proof we will show that with probability at least $1 - n^{-2}$ it holds that $\sum_{(i,j)\in R}b_{ij} \le O(s\varepsilon) \le c\lg n/\lg(1/\varepsilon)$.
Since $\supp(x_{head(\ell)})$ is well-behaved, there are at most $6\lg n/\lg(1/\eps)$ rows $i \in [m]$ for which $|u_i| > \sqrt{5/s}$ and since $\|x_{tail(\ell)}\|=1$ there are at most $\ell \lg^4(1/\eps)$ columns $j \in \supp(x_{tail(\ell)})$ such that $|x_j| \ge 1/(\sqrt{\ell}\lg^2(1/\varepsilon))$. Thus $|R| \le 6\ell \lg n \lg^3(1/\eps)$, and therefore $\mu : = \mathbb{E}\left[\sum_{(i,j)\in R}b_{ij}\right] \le (s/m) \cdot 6\ell \lg n \lg^3(1/\eps) \le \varepsilon^{1/2}\lg n \lg^3(1/\eps)$, where the last inequality follows from the fact that $s \le \varepsilon m$ and $\ell \le \varepsilon^{-1/2}$. For $\varepsilon$ smaller than some constant we get that the expectation is at most $\mu \le \varepsilon^{1/4}\lg n / \lg(1/\eps)$. Straightforward calculations give the following observation, whose proof is deferred to the appendix.
\begin{observation}
\label{obs:withWithoutReplacement}
For every $t>0$, $\mathbb{E}\left[\exp\left(t\sum_{(i,j)\in R}b_{ij}\right)\right] \le \prod_{(i,j)\in R}{\mathbb{E}\left[\exp\left(tb_{ij}\right)\right]}$.
\end{observation}
Employing Observation~\ref{obs:withWithoutReplacement} we can apply Hoeffding-like inequalities on the probability that $\sum_{(i,j)\in R}b_{ij}$ is large. Specifically for a large enough constant $c$ let $\delta = c \varepsilon^{-1/4}-1$ and $t=\ln(1+\delta)$ we get from Markov's inequality that
\begin{equation*}
\Pr\left[\sum_{(i,j)\in R}b_{ij} > \frac{c \lg n}{\lg (1/\varepsilon)}\right] = \Pr\left[\exp\left(t\sum_{(i,j)\in R}b_{ij}\right) > \exp\left(\frac{tc \lg n}{\lg (1/\varepsilon)}\right)\right] \le \frac{e^{\delta\mu}}{(1+\delta)^{c \lg n / \lg(1/\varepsilon)}} \;.
\end{equation*}
As $(1+\delta) = c\varepsilon^{-1/4}$ and $\mu \le \varepsilon^{1/4}\lg n / \lg(1/\eps)$ we get that if $c$ is large enough
\begin{equation*}
\Pr\left[\sum_{(i,j)\in R}b_{ij} > \frac{c \lg n}{\lg (1/\varepsilon)}\right] \le \left(e\varepsilon^{1/4}\right)^{c \lg n / \lg(1/\varepsilon)} \le n^{-2}\;.
\end{equation*}
\end{proof}

\begin{claim}\label{c:crossSmallTerms}
$\Pr\left[\left|\sum\limits_{(i,j)\in S}u_i \cdot b_{ij} \sigma_{ij} x_j/\sqrt{s}\right| \le O(\varepsilon) \right] \ge 1-O(n^{-2})$.
\end{claim}

\begin{proof}
We first note that the sum can be thought of as an inner product between two vectors indexd by $(i,j)\in S$. Specifically let $\sigma, w \in \mathbb{R}^S$ be defined as follows. For every $(i,j)\in S$, $\sigma_{(i,j)} = \sigma_{ij}$ and $w_{(i,j)} = c_{(i,j)}b_{ij}$, where $c_{(i,j)}= u_ix_j/\sqrt{s}$. As $\sigma$ and $w$ are independent, we get from Hoeffding's inequality that for every $c>0$
\[
\Pr\left[|\langle w, \sigma\rangle | > c\varepsilon \mid \|w\| \right] \leq 2\exp\left(-\frac{(c\varepsilon)^2}{2\|w\|^2}\right) \;.
\]
Therefore it is enough to show that with probability at least $1-O(n^{-2})$ it holds that $\|w\|^2 = O(\varepsilon^2/\lg n)$. Note first that
\[
\mathbb{E}\left[\|w\|^2\right] = \mathbb{E}\left[\sum_{(i,j)\in S}c_{(i,j)}^2b_{ij}^2\right]= \frac{1}{s}\sum_{(i,j)\in S}u_i^2x_j^2\mathbb{E}[b_{ij}] = \frac{1}{m}\sum_{(i,j)\in S}u_i^2x_j^2 = \frac{1}{m}\|x_{tail(\ell)}\|^2\|u\|^2 \;.
\]
Since we conditioned on $\|u\|^2 < 2$, and since $\|x_{tail(\ell)}\|^2 \le 1$ we have that $\mathbb{E}\left[\|w\|^2\right] \le 2/m = O(\varepsilon^2/\lg n)$.
Our goal is therefore to bound $Pr[\|w\|_2 > (1+\delta)\mathbb{E}[\|w\|^2]]$ for some constant $\delta>0$.
Similarly to the previous proof we employ the following observation, whose proof is deferred to the appendix.
\begin{observation}
\label{obs:withWithoutReplacementCoeff}
For every $t>0$, $\mathbb{E}\left[\exp\left(t\sum_{(i,j)\in S}c_{(i,j)}^2b_{ij}\right)\right] \le \prod_{(i,j)\in S}{\mathbb{E}\left[\exp\left(tc_{(i,j)}^2b_{ij}\right)\right]}$.
\end{observation}
We start by bounding the coefficients $c_{(i,j)}$. Recall that $\|u\|_\infty \le \sqrt{\ell/s}$ and $\|x_{tail(\ell)}\|_\infty \le 1/\sqrt{\ell}$, and let $(i,j)\in S$. Then either $|u_i| \le \sqrt{5/s}$ or $|x_j|\le 1/(\sqrt{\ell}\lg^2(1/\varepsilon))$. In the former case $|u_ix_j/\sqrt{s}| \le \sqrt{5}/(s\sqrt{\ell})$, and by the choice of $s$ and $\ell$ we get $|u_ix_j/\sqrt{s}| \le O(\varepsilon/\lg n)$. In the latter case $|u_ix_j/\sqrt{s}| \le 1/s\lg(1/\varepsilon) = O(\varepsilon/\lg n)$. We conclude that for all $(i,j)\in S$ we have $|c_{(i,j)}|=|u_ix_j/\sqrt{s}| \le O(\varepsilon/\lg n)$.
Let $\mu = \mathbb{E}[\|w\|^2]$, $\alpha = O((\varepsilon / \lg n)^2)$ and let $t=\ln(1+\delta)/\alpha$ for some large enough constant $\delta$, then we get from Markov's inequality that
\begin{equation}
\Pr\left[\|w\|^2 > (1+\delta)/m \right]  \le \frac{\mathbb{E}\left[\exp\left(t\sum\limits_{(i,j)\in S} c_{(i,j)}^2b_{ij}\right)\right]}{\exp(t(1+\delta)/m)}\le \frac{\prod\limits_{(i,j)\in S}\mathbb{E}\left[\exp\left(tc_{(i,j)}^2b_{ij}\right)\right]}{(1+\delta)^{(1+\delta)/(\alpha m)}}
\label{eq:Markov}
\end{equation}
Now note that for every $(i,j)\in S$ it holds that  
\[
\mathbb{E}\left[\exp\left(t c_{(i,j)}^2b_{ij}\right)\right] = \frac{s}{m}e^{t c_{(i,j)}^2} + \left(1-\frac{s}{m}\right) = 1 + \frac{s}{m} \left(e^{t c_{(i,j)}^2} - 1\right) = 1 + \frac{s}{m} \left((1+\delta)^{c_{(i,j)}^2/\alpha} - 1\right)
\]
Since $c_{(i,j)}^2 \le \alpha$, we get that $(1+\delta)^{c_{(i,j)}^2/\alpha} \le 1 + \delta c_{(i,j)}^2/\alpha$. Therefore 
\[
\mathbb{E}\left[\exp\left(t c_{(i,j)}^2b_{ij}\right)\right] \le 1 + \frac{sc_{(i,j)}^2\delta}{\alpha m} \le \exp\left(\frac{\delta}{\alpha} \cdot \frac{s c_{(i,j)}^2}{m}\right)\;.
\] Plugging into \eqref{eq:Markov} we get that
\begin{equation*}
\Pr\left[\|w\|^2 > (1+\delta)\mu \right] \le \frac{\prod\limits_{(i,j)\in S}\exp\left(\frac{\delta}{\alpha} \cdot \frac{s c_{(i,j)}^2}{m}\right)}{(1+\delta)^{(1+\delta)/(\alpha m)}} = \frac{e^{\delta\mu/\alpha}}{(1+\delta)^{(1+\delta)/(\alpha m)}} \le \left(\frac{e^{2\delta}}{(1+\delta)^{1+\delta}}\right)^{1/(\alpha m)}\;,
\end{equation*}
where the last inequality is due to the fact that $\mu \le 2/m$. As $2/(\alpha m) = \Omega (\lg n)$, then for a large enough constant $\delta$  the probability is at most $n^{-2}$
\end{proof}

\section{Sparsity Lower Bound}
In this section, we prove our lower bound result, Theorem~\ref{thm:lower}. Let $0 < \eps < 1/4$. We first define a hard set of input vectors in $\R^d$. Let $\ell = \lg n/\lg (ed/\ell)$. For every $\ell$-sized subset $S \subseteq [d]$ of coordinates, form the vector $x_S = \sum_{i \in S} e_i/\sqrt{\ell}$. The collection of these vectors, along with the $0$-vector and $e_1,\dots,e_d$, is our hard input instance $X$ of cardinality $|X| \leq \binom{d}{\ell} + 1 + d \leq (e d/\ell)^\ell + n \leq 2n$.

Assume that $A$ is an $m \times d$ matrix in which every column has at most $s$ non-zeros, and that $A$ satisfies $\|Au-Av\|^2 \in (1 \pm \eps)\|u-v\|^2$ for all $u,v \in X$. We also assume that $m =\Omega(\eps^{-2} \lg n)$ as such a lower bound on $m$ is already known. We prove a lower bound on $s$ from these assumptions. Throughout the proof, we assume $s \leq m/2$ as otherwise, we are already done.

Let $a_j$ denote the $j$'th column of $A$. We first observe that $\|a_j\|^2 \in (1 \pm \eps)$ for all $j$ since $\|a_j\|^2 = \|Ae_j\|^2 = \|Ae_j - A0\|^2 \in (1 \pm \eps)\|e_j - 0\|^2 = (1 \pm \eps)$.

Our next step is to identify a subset $T \subseteq [m]$, such that many of the columns of $A$ have large entries in $T$. For this, we prove the following lemma:

\begin{lemma}
  \label{lem:manyheavy}
  Let $v \in \R^m$ be a vector with at most $s \leq m/2$ non-zeros. For any $t \leq s/8$, there are at least $\min\{\binom{m-1}{t-1}, (s/(8t))^t\}$ distinct subsets $T \subseteq [m]$ of cardinality $|T|=t$ for which $\sum_{i \in T} v_i^2 \geq t\|v\|^2/(2s)$.
\end{lemma}

We defer the proof to the end of the section and instead proceed with the lower bound argument.

Let $t$ be a parameter to be fixed. There are $d$ columns in $A$, which by Lemma~\ref{lem:manyheavy} and averaging among all $t$-sized subsets of $[m]$ implies that there is a $T$ with $|T|=t$ such that at least $d \min\{\binom{m-1}{t-1}, (s/(8t))^t\}/\binom{m}{t} \geq d \min\{t/m, (s/(8t))^t/(em/t)^t\} = d \min\{t/m, (s/(8e m))^t \}$ columns $a_j$ of $A$ satisfy $\sum_{i \in T} a_{i,j}^2 \geq (1-\eps)t/(2s) \geq t/(4s)$. Fix such a $T$ and let $A_T$ be the subset of columns satisfying the previous conditions for this $T$.

Let $\fnet$ be a $1/4$-net for the set of unit vectors in $\R^t$, i.e. for any $x \in \R^t$ with $\|x\|=1$, there is an $x' \in \fnet$ with $\|x-x'\| \leq 1/4$ and $\|x'\|=1$. Standard results give that there is such a $\fnet$ of cardinality $2^{O(t)}$. For every $a_j \in A_T$, let $a^T_j$ denote $a_j$ restricted to the $t$ entries in $T$ and let $w(a_j)$ denote the closest vector in $\fnet$ to $a^T_j/\|a^T_j\|$. By averaging, there is a vector $w \in \fnet$ where at least $d \min\{t/m, (s/m)^t\}2^{-O(t)}$ vectors $a_j \in A_T$ have $w$ as the closest vector to $a^T_j/\|a^T_j\|$. Fix such a $w$ and let $A_{T,w}$ be the subset of columns in $A$ satisfying the conditions.

Now fix $t = (1-o(1))\lg(\eps d/\ell)/\lg(m/s)$. Assume first that for this choice, we have $(s/m)^t \leq t/m$.
Then since $s =o(m)$ (otherwise we are done with the lower bound proof), we have
\[
  |A_{T,w}| \geq d (s/m)^t 2^{-O(t)} = d (s/m)^{(1+o(1))t} \geq \ell/\eps.
\]
From $A_{T,w}$, construct $\eps^{-1}$ disjoint sets of $\ell$ vectors each. For each such set $S$, we have that the vector $\sum_{a_j \in S}e_j/\sqrt{\ell}$ is in $X$. Let $v_1,\dots,v_{\eps^{-1}}$ denote these vectors. Since they have disjoint supports, we have $\langle v_i, v_j\rangle =0$ for $i \neq j$ and thus $\|v_i - v_j\|^2 = 2$. This also implies that $\|Av_i - Av_j\|^2 \in 2 \pm 2\eps$. Since $\|Av_i - Av_j\|^2 = \|Av_i\|^2 + \|Av_j\|^2 - 2\langle Av_i,Av_j\rangle$ and $\|Av_i\|^2, \|Av_j\|^2 \in 1 \pm \eps$, it must be the case that $\langle A v_i, Av_j \rangle \in \pm 2\eps$.

On the other hand, we have $\langle Av_i , A v_j \rangle = \sum_{a_h \in S_i} \sum_{a_k \in S_j} \langle a_h, a_k \rangle/\ell$. Thus $\sum_{a_h \in S_i} \sum_{a_k \in S_j} \langle a_h, a_k \rangle/\ell \leq 2 \eps$. Now set all entries $i \in T$ to $0$ for all columns of $A$. Call the resulting columns $\hat{a}_j$ and the resulting matrix $\hat{A}$. Then for two columns $a_h, a_k$, we have $\langle \hat{a}_h, \hat{a}_k\rangle = \langle a_h ,a_k\rangle - \langle a^T_h, a^T_k \rangle$. For any two $a_h ,a_k \in A_{T,w}$, we have $\langle a^T_h, a^T_k \rangle = \|a^T_h\| \|a^T_k\| \langle w + (a^T_h/\|a^T_h\|-w), w + (a^T_k/\|a^T_k\|-w) \rangle$. We have
\begin{eqnarray*}
  \langle w + (a^T_h/\|a^T_h\|-w), w + (a^T_k/\|a^T_k\|-w) \rangle &=& \\
  \|w\|^2 + \langle w,  (a^T_h/\|a^T_h\|-w) \rangle + \langle w,  (a^T_k/\|a^T_k\|-w) \rangle + \langle (a^T_h/\|a^T_h\|-w), (a^T_k/\|a^T_k\|-w)\rangle &\geq& \\
  \|w\|^2 - \|w\| \|a^T_h/\|a^T_h\|-w \| - \|w\| \|a^T_k/\|a^T_k\|-w \| - \|a^T_h/\|a^T_h\|-w \||a^T_h/\|a^T_h\|-w \|
                                                                   &\geq& \\
  1 - 1/4 - 1/4 - 1/16 &\geq& \\
  1/4.
\end{eqnarray*}
Hence for any two $a_h ,a_k \in A_{T,w}$, it holds that $\langle \hat{a}_h, \hat{a}_k\rangle \leq \langle a_h, a_k\rangle - \|a^T_h\| \|a^T_k\|/4 \leq \langle a_h, a_k\rangle - t(1-\eps)/(8s) \leq \langle a_h, a_k\rangle - t/(16s)$. We therefore conclude that $\langle \hat{A}v_i, \hat{A}v_j \rangle \leq \sum_{a_h \in S_i} \sum_{a_k \in S_j} (\langle a_h, a_k \rangle - t/(16s))/\ell \leq 2 \eps - \ell t/(16 s)$.

Finally, consider the vector $z=\sum_{i =1}^{\eps^{-1}}\hat{A}v_i$. We have $\|z\|^2 = \sum_{i=1}^{\eps^{-1}} \|\hat{A}v_i\|^2 + \sum_{i \neq j} \langle \hat{A}v_i, \hat{A}v_j \rangle \leq \eps^{-1}(1+\eps) + \eps^{-1} (\eps^{-1}-1) (2 \eps - \ell t/(16 s))$. Since $\|z\|^2 \geq 0$, it must thus be the case that $ (\eps^{-1}-1) \ell t/(16 s) \leq 1 + \eps + (\eps^{-1}-1)2 \eps$. Since $\eps^{-1}-1 \geq \eps^{-1}/2$ and $1 + \eps + (\eps^{-1}-1)2 \eps \leq 4$, we conclude $s \geq \eps^{-1} \ell t/128$. Since $d \geq m = \Omega(\eps^{-2} \lg n)$, we have $\lg(ed/\ell) \leq c'\lg(\eps d/\ell)$ for a constant $c'>0$. Thus
\[
  s = \Omega(\eps^{-1} \lg n/\lg(m/s)) = \Omega(\eps^{-1} \lg n/\lg(m/\lg n)).
\]
This was only under the assumption that $(s/m)^t \leq t/m$ for $t=(1-o(1)) \lg(\eps d/\ell)/\lg(m/s)$. This is implied by $(\ell/(\eps d))^{1-o(1)} \leq (1-o(1)) \lg(\eps d/\ell)/(m \lg(m/s))$. This is in particular implied by $m \leq (\eps d/\ell)^{1-o(1)}$. Constraining $m \leq (\eps d/\lg n)^{1-o(1)}$ thus completes the proof.

\begin{proof}[Proof of Lemma~\ref{lem:manyheavy}]
First consider the case where $v$ has at least one coordinate $j$ with $v_j^2 \geq t\|v\|^2/(2s)$. In this case, there are at least $\binom{m-1}{t-1}$ valid choices for $T$.

If all coordinates $j$ satisfy $v_j^2 < t \|v\|^2/(2s)$, we partition the coordinates of $v$ into buckets based on their magnitude. Concretely, for every $i=0,\dots,\lg t-1$, let $V_i$ denote the subset of coordinates $j$ for which $v_j^2 \in [\|v\|^22^{i-1}/s,\|v\|^2 2^{i}/s)$. Notice that all coordinates of $j$ with $v_j^2 < \|v\|^2/(2s)$ contribute at most $s \|v\|^2/(2s) = \|v\|^2/2$ to $\|v\|^2$. Furthermore, the contribution from coordinates $j$ with $j \in V_i$ for a $V_i$ with $|V_i| \leq s/(4t)$, is no more than $\sum_{i=0}^{\lg t-1} (s/(4t))\|v\|^2 2^i/s \leq \|v\|^2/4$. Hence $\sum_{i : |V_i|>s/(4t)} \sum_{j \in V_i} v_j^2 \geq \|v\|^2/4$. This implies that we also have $\sum_{i : |V_i|>s/(4t)} |V_i| \cdot \|v\|^22^i/s \geq \|v\|^2/4$.

For each $i$ with $|V_i| > s/(4t)$, let $t_i := \lceil 4t|V_i|/s \rceil$. Then $t_i \leq 4t|V_i|/s + 1 \leq 4t|V_i|/s + 4t|V_i|/s \leq 8t|V_i|/s$.
Consider all sets $T$ having $|T \cap V_i|=t_i$ for all $i$ with $|V_i| > s/(4t)$. Any such $T$ satisfies $\sum_{j \in T} v_j^2 \geq \sum_{i : |V_i|>s/(4t)}t_i\|v\|^2 2^{i-1}/s \geq  (2t/s)\sum_{i :
 |V_i|>s/(4t)}|V_i| \cdot \|v\|^2 2^{i}/s  \geq (t/(2s)) \|v\|^2$. The number of such $T$ is at least $\binom{m/2}{t-\sum_i t_i}\prod_{i : |V_i|>s/(4t)} \binom{|V_i|}{t_i}$. For $|V_i| > s/(4t)$, we have $\binom{|V_i|}{t_i} \geq (|V_i|/t_i)^{t_i} \geq (|V_i|/(8t|V_i|/s))^{t_i} = (s/(8t))^{t_i}$. The number of valid $T$ is thus at least $\binom{m-s}{t-\sum_i t_i} \prod_{i : |V_i|>s/(4t)} (s/(8t))^{t_i} \geq (m/(2t))^{t - \sum_i t_i} (s/(8t))^{\sum_i t_i} \geq (s/(8t))^t$.
\end{proof}

\section{Subspace Embeddings}
In this section, we show that for any $k$-dimensional subspace $V \subset \R^d$, an embedding matrix $A$ sampled as in Kane and Nelson~\cite{KN14}, with a sparsity $s = \Theta(\eps^{-1}(k/\lg(1/\eps) + k^{2/3} \lg^{1/3} k))$ as in Theorem~\ref{th:subspacesMain}, preserves the norm of every vector in $V$ to within $1 \pm \eps$ with high probability, thus proving Theorem~\ref{th:subspacesMain}.

To simplify the proof, we will once again argue that norms are preserved to within $1 \pm O(\eps)$. As in Section~\ref{sec:upper}, simple rescaling of $\eps$ by a constant factor implies the result.

We first show that it is enough that $A$ approximately preserves norms of a finite set defined by a $1/2$-net on the subspace. The following lemma is known and appears in previous works. For sake of completeness, we supply a proof, which is deferred to the appendix, Section~\ref{sec:net}.
\begin{lemma}
\label{lem:halfnet}
Let $A$ be a matrix and $V$ a subspace of $\R^d$. Let $\enet$ be a $1/2$-net for $V$ and $\enet^+ = \{x + y : x,y \in \enet \cup \{0\}\}$.  Assume that for all $v \in \enet^+$, $\|Av\|^2 \in (1\pm O(\varepsilon))\|v\|^2$, then for all unit vectors $x \in V$, $\|Ax\|^2 \in (1\pm O(\varepsilon))\|x\|^2$.
\end{lemma}

As explained in the technical overview, we also employ Lemma~\ref{lem:subspacecover}. The lemma gives a combinatorial property of subspaces of $\mathbb{R}^d$.
\begin{customlem}{\ref{lem:subspacecover}}
  Let $V$ be a $k$-dimensional subspace of $\R^d$. For every $\ell \ge 1$, there is a set $S \subseteq [d]$ of coordinates with $|S| \leq k \ell$ such that for every unit vector $v \in V$, all coordinates $i \in [d] \setminus S$ satisfy $|v_i| < 1/\sqrt{\ell}$.
\end{customlem}

\begin{proof}
Let $v^1,\dots,v^k$ be an orthonormal basis for $V$. Consider any unit vector $u \in V$ and write it as $u = \sum_j \alpha_j v^j$ with $\sum_j \alpha_j^2 =1$. Then $u_i = \sum_j \alpha_j v^j_i$. By Cauchy-Schwarz, we have $|u_i| \leq \sqrt{\sum_j \alpha_j^2} \cdot \sqrt{\sum_j (v^j_i)^2} = \sqrt{\sum_j (v^j_i)^2}$. Now let $S \subseteq [d]$ be all coordinates such that there is a unit vector $u \in V$ with $|u_i| \geq 1/\sqrt{\ell}$. Then for all $i \in S$, we must have $\sum_j (v^j_i)^2 \geq 1/\ell$. But $\sum_j \sum_i (v^j_i)^2 = k$ and thus $|S| \leq k \ell$.
\end{proof}

With the two lemmas above, we are ready to prove our main result on subspace embeddings, captured in Theorem~\ref{th:subspacesMain}. Similarly to the proof of Theorem~\ref{th:upperBoundMain}, we define the following notation.

\begin{notation}
Let $x \in \mathbb{R}^{d}$. For every $\ell \in [d]$ denote by $x_{heavy(\ell)}$ the vector obtained from $x$ where all but the entries of magnitude strictly greater than $1/\sqrt{\ell}$ are zeroed out. Denote $x_{light(\ell)} = x - x_{heavy(\ell)}$.
\end{notation}

Let $\enet$ be a $1/2$-net on the unit ball in $V$, and define $\enet^+ = \enet \cup \{x+y : x,y \in \enet \cup \{0\}\}$. A $1/2$-net can be constructed such that $|\enet| \le 4^k$. Let $n = |\enet^+| \le 8^k$.
Let $\ell = \left\lceil \min\left\{\eps^{-1/2}, \left(\frac{\lg n}{\lg k}\right)^{2/3} \right\} \right\rceil$ be an integer, and let $S$ be defined as in Lemma~\ref{lem:subspacecover}. We define $\Y$ as the set of all vectors $y \in \R^d$ such that $\supp(y) \subseteq S$, $|\supp(y)| \le \ell$ and $\|y\| \le 1$. 

Define ${\cal E}_1$ to be the set of all matrices $A \in \mathbb{R}^{m\times d}$ such that for all $x \in \Y$, $\|Ax\|^2 \in (1 \pm \varepsilon)\|x\|^2$.
Define ${\cal E}_2$ to be the set of all matrices $A \in \mathbb{R}^{m\times d}$ such that for all $x \in \enet^+$, $\|Ax_{light(\ell)}\|^2 \in \|x_{light(\ell)}\|^2 \pm \varepsilon$.
Define ${\cal E}_3$ to be the set of all matrices $A \in \mathbb{R}^{m\times d}$ such that for all $x \in \enet^+$, $\left|\dotp{Ax_{heavy(\ell)}}{Ax_{light(\ell)}}\right| < \varepsilon$.

\begin{claim}
Assume $A \in {\cal E}_1 \cap {\cal E}_2 \cap {\cal E}_3$. Then for every unit vector $x \in V$, $\|Ax\|^2 \in (1 \pm O(\varepsilon))$.
\end{claim}

\begin{proof}
Following Lemma~\ref{lem:halfnet} and using linearity of $A$, it is enough to prove the claim for $x = z/\|z\|$ for all vectors $z$ in $\enet^+$. Let therefore $x$ be any such unit vector. Then $\|Ax\|^2 = \|Ax_{heavy(\ell)}\|^2 + \|Ax_{light(\ell)}\|^2 + 2 \dotp{Ax_{heavy(\ell)}}{Ax_{light(\ell)}}$. Since $\|x\|=1$ and every entry of $x_{heavy(\ell)}$ is at least of magnitude $1/\sqrt{\ell}$, we have by the definition of $S$ that $\supp(x_{heavy(\ell)}) \subseteq S$ and $|\supp(x_{heavy(\ell)})| \le \ell$. Therefore $x_{heavy(\ell)} \in \Y$ and thus
\[
\|Ax\|^2 \le (1+\varepsilon)\|x_{heavy(\ell)}\|^2 + (1 + \eps)\|x_{light(\ell)}\|^2 + \varepsilon + 2 \varepsilon \le (\|x\|^2+O(\varepsilon))
\]
Similarly
\[
\|Ax\|^2 \ge (1-\varepsilon)\|x_{heavy(\ell)}\|^2 + (1-\eps)\|x_{light(\ell)}\|^2 -\varepsilon - 2 \varepsilon \ge (\|x\|^2-O(\varepsilon))
\]
\end{proof}
As in the proof of Theorem~\ref{th:upperBoundMain}, it remains to lower bound the probability of ${\cal E}_1 \cap {\cal E}_2 \cap {\cal E}_3$. Once again, we bound the probability of each event separately.
\begin{claim}
$\Pr[A \in {\cal E}_1] \ge 1 - 2^{-k^{2/3}}$.
\end{claim}

\begin{proof}
Denote $\delta = 2^{-\Omega(\sqrt[3]{\lg^2 n \cdot \lg k})}$. We get that $m \ge \Omega(\varepsilon^{-2}\lg (1/\delta))$ and $s \ge \Omega(\varepsilon^{-1}\lg (1/\delta))$. Following the result by Kane and Nelson \cite{KN14}, for every unit vector $x \in \mathbb{R}^d$, we have that $\Pr\left[\|Ax\|^2 \in (1\pm O(\varepsilon))\right] \ge 1 - \delta$. 

Next, for every $T \subseteq S$ such that $|T|=\ell$, let $\Y_T = \{y \in \R^d : \|y\| \le 1 \;\; and \;\; \supp(y)\subseteq T\}$, then $\Y_T$ is a unit ball of an $\ell$-dimensional subspace of $\R^d$, and thus there is a $1/2$-net $\hat{\Y}_T$ for $\Y_T$ such that $|\hat{\Y}_T| \le 4^{\ell}$. Note that in these notations $\Y = \bigcup_{T \in \binom{S}{\ell}}\Y_T$, and denote in addition $\hat{\Y}= \bigcup_{T \in \binom{S}{\ell}}\hat{\Y}_T$. Then
\[
|\hat{\Y}| \le \binom{|S|}{\ell}4^{\ell} \le (4ek)^{\ell} = 2^{\Omega(\ell \lg(4ek))}\;.
\]
For $k>1$, we have $|\hat{\Y}| \le 2^{\Omega(\ell \lg k)} = 2^{\Omega(\sqrt[3]{\lg^2 n\lg k})} = \delta^{-1/2}$.
Therefore with probability at least $1 - \sqrt{\delta} \ge 1 - 2^{-k^{2/3}}$, we get that for all $y \in \hat{\Y}$, $\|Ay\|^2 \in  (1\pm O(\varepsilon))$.

Assume therefore that for all $y \in \hat{\Y}$, $\|Ay\|^2 \in  (1\pm O(\varepsilon))$. Let $x \in \Y$, then there exists $T \subseteq S$ such that $|T|=\ell$ and $\supp(x) \subseteq T$, hence $x \in \Y_T$. As $\hat{\Y}_T$ is a $1/2$-net of $\Y_T$ and $\Y_T$ is a unit ball of an $\ell$-dimensional subspace of $\R^d$, Lemma~\ref{lem:halfnet} implies that $\|Ax\|^2 \in (1 \pm O(\varepsilon))\|x\|^2$. Therefore $\Pr[A \in {\cal E}_1] \ge 1 - 2^{-k^{2/3}}$.
\end{proof}

The following claim completes the proof of Theorem~\ref{th:subspacesMain}. Proving bounds on the probabilities of ${\cal E}_2$ and ${\cal E}_3$ is analogous to the proofs of Claims~\ref{c:tails} and \ref{c:cross} respectively, and is therefore omitted.

\begin{claim}
$\Pr[A \in {\cal E}_2] \ge 1 - \frac{1}{n}$ and $\Pr[A \in {\cal E}_3] \ge 1 - \frac{1}{n}$.
\end{claim}

\bibliographystyle{alphaurlinit}
\bibliography{sparserJLBib}

\appendix	

\section{Proofs for Observations~\ref{obs:withWithoutReplacement} and \ref{obs:withWithoutReplacementCoeff}}
For sake of completeness, we prove the following lemma, which implies Observations~\ref{obs:withWithoutReplacement} and \ref{obs:withWithoutReplacementCoeff}.

\begin{lemma}
Let $I \subseteq [m]\times[d]$ and let $\{c_{(i,j)}\}_{(i,j)\in I}$ be a set of non-negative constants. For every $(i,j)\in I$, define $b_{ij}$ as the Bernoulli random variable attaining $1$ if and only if $a_{ij}\ne 0$, then 
\[
\mathbb{E}\left[\exp\left(\sum_{(i,j)\in I}c_{(i,j)}b_{ij}\right)\right] \le \prod_{(i,j)\in I}\mathbb{E}[\exp(c_{(i,j)}b_{ij})]
\]
\end{lemma}

\begin{proof}
Recall that the rows of $A$ are divided into $s$ blocks $I_1,\ldots,I_s$ of $m/s$ consecutive rows each. That is for every $p \in [s]$, $I_p = [(p-1)(m/s)+1, pm/s]$. In these notations, 
\[
\sum_{(i,j)\in I}c_{(i,j)}b_{ij} = \sum_{j \in [d]}\sum_{p\in S}\sum_{i \in I_p : (i,j)\in I}c_{(i,j)}b_{ij} \;.
\]
As the columns of $A$, as well as different blocks within each column are independent, we get that
\[
\mathbb{E}\left[\exp\left(\sum_{(i,j)\in I}c_{(i,j)}b_{ij}\right)\right] \le \prod_{j \in [d]}\prod_{p \in [s]}\mathbb{E}\left[\prod_{i \in I_p : (i,j)\in I}\exp(c_{(i,j)}b_{ij})\right]
\]
Fix $j \in [d]$ and $p \in [s]$, and denote $C = \{i \in I_p : (i,j)\in I\}$. For every $i \in C$,  $c_{(i,j)} \ge 0$, and thus $\frac{e^{c_{i,j}}-1}{(m/s)} \ge 0$. Therefore
\begin{equation*}
\begin{split}
\mathbb{E}\left[\prod_{i \in C}\exp(c_{(i,j)}b_{ij})\right] &= \sum_{i \in C}\frac{e^{c_{(i,j)}}}{(m/s)} + (1 - \frac{|C|}{(m/s)}) = 1 + \sum_{i \in C}\frac{e^{c_{(i,j)}} - 1}{(m/s)} \\
&\le \prod_{i \in C} \left(1 + \frac{e^{c_{(i,j)}}-1}{(m/s)}\right) = \prod_{i\in C} \mathbb{E}(\exp(c_{(i,j)}b_{ij}))
\end{split}
\end{equation*}
\end{proof}

\section{A $1/2$-net suffices}
\label{sec:net}
Here we give the defered proof of Lemma~\ref{lem:halfnet}

\begin{proof}[Proof of Lemma~\ref{lem:halfnet}]
Let $x \in V$ be a unit vector. We construct inductively a sequence $\{x_i\}_{i=0}^\infty$ of vectors in $\enet$ and a sequence $\{\alpha_i\}_{i=0}^\infty$ of non-negative real numbers such that $x = \sum_{i=0}^\infty{\alpha_ix_i}$ and moreover $\alpha_i \le 2^{-i}$ for all $i \ge 0$.
Let $x_0$ be the closest vector to $x$ in $\enet$, and let $\alpha_0=1$. Then $x = \alpha_0x_0 + (x-\alpha_0x_0)$. Clearly if $x-\alpha_0x_0 = 0$ we are done, as we can define $\alpha_i =0$ for all $i \ge 1$. Otherwise, denote $\alpha_1 = \|x - \alpha_0x_0\|$ and $v_1 = \alpha_1^{-1}(x-\alpha_0x_0)$, then $\alpha_1 \le 1/2$, $v_1$ is a unit vector and $x = \alpha_0x_0 + \alpha_1v_1$. 
Following by induction let $p \in \mathbb{N}$ and assume there are vectors $x_0,\ldots,x_p \in \enet$, numbers $\alpha_0,\ldots,\alpha_{p+1}$ and a unit vector $v_{p+1}$ such that $x = \sum_{i=0}^p\alpha_ix_i + \alpha_{p+1}v_{p+1}$ and such that $\alpha_i \le 2^{-i}$ for all $i \le p+1$. Let $x_{p+1}$ be the closest vector in $\enet$ to $v_{p+1}$. Then $v_{p+1} = x_{p+1} + (v - x_{p+1})$.
If $v - x_{p+1} = 0$ we are done, as we can define $\alpha_i =0$ for all $i \ge p+2$. Otherwise, denote $\beta = \|v - x_{p+1}\|$, $\alpha_{p+2}  = \alpha_{p+1}\beta$  and $v_{p+2} = \beta^{-1}(v - x_{p+1})$, then $\alpha_{p+2} \le 2^{-p+1}$, $v_{p+2}$ is a unit vector and 
\[
x = \sum_{i=0}^p\alpha_ix_i + \alpha_{p+1}v_{p+1} = \sum_{i=0}^p\alpha_ix_i + \alpha_{p+1}(x_{p+1} + (v - x_{p+1})) = \sum_{i=0}^{p+1}\alpha_ix_i + \alpha_{p+2}v_{p+2} \;.
\]
This completes the construction of the sequences. Next note that 
\[
\|x\|^2 = \left\|\sum_{i=0}^\infty{\alpha_ix_i}\right\|^2 = \sum_{i=0}^\infty{\alpha_i \left\| x_i \right\|^2} + \sum_{i < j}2 \alpha_i\alpha_j x_i^tx_j  \;.
\]
Similarly we get that 
\[
\|Ax\|^2 = \left\|\sum_{i=0}^\infty{\alpha_iAx_i}\right\|^2 = \sum_{i=0}^\infty{\alpha_i \left\| Ax_i \right\|^2} + \sum_{i < j}2\alpha_i\alpha_j x_i^tA^tAx_j  \;.
\]
Since $x_i \in \enet \subseteq \enet^+$ for all $i$ we have that $\left\| Ax_i \right\|^2 \in 1 \pm O(\varepsilon)$. In addition, for all $i<j$, $2x_i^tx_j = \|x_i+x_j\|^2 - \|x_i\|^2 - \|x_j\|^2$. Since $x_i,x_j,x_i+x_j \in \enet^+$ we have that
\[
2x_i^tA^tAx_j = \|Ax_i+Ax_j\|^2 - \|Ax_i\|^2 - \|Ax_j\|^2 = \|A(x_i+x_j)\|^2 - \|Ax_i\|^2 - \|Ax_j\|^2 \in 2x_i^tx_j \pm O(\varepsilon) \;,
\]
and thus
\begin{equation*}
\begin{split}
\|Ax\|^2 &= \sum_{i=0}^\infty{\alpha_i \left\| Ax_i \right\|^2} + \sum_{i < j}2\alpha_i\alpha_j x_i^tA^tAx_j \\
&\in \sum_{i=0}^\infty{\alpha_i (\|x_i\|^2 \pm O(\varepsilon))} + \sum_{i < j}2 \alpha_i\alpha_j (x_i^tx_j \pm O(\varepsilon))\\
&\subseteq \sum_{i=0}^\infty{\alpha_i \left\| x_i \right\|^2} + \sum_{i < j}2 \alpha_i\alpha_j x_i^tx_j + O(\varepsilon)\left(\sum_{i=0}^\infty{\alpha_i} + \sum_{i < j}2 \alpha_i\alpha_j\right) \subseteq 1 \pm O(\varepsilon)
\end{split}
\end{equation*}
\end{proof}

\end{document}